\documentclass{article}

  \usepackage{underscore}         
  \usepackage[T1]{fontenc}        
\usepackage{amsmath}
\usepackage{amssymb}
\usepackage{amsthm}
\usepackage{xcolor,color}
\usepackage{hyperref}
\usepackage[capitalise]{cleveref}
\usepackage{stmaryrd}
\usepackage{wasysym}

\newcommand{\ra}{\rightarrow}
\newcommand{\eq}{\leftrightarrow}

\newcommand{\sneg}{{\sim}}

\newcommand{\ocirc}{{\fullmoon}} 
\newcommand{\bcirc}{{\newmoon}}

\newcommand{\aland}{\tilde{\land}}
\newcommand{\alor}{\tilde{\lor}}
\newcommand{\ato}{\tilde{\to}}
\newcommand{\aneg}{\tilde{\neg}}

 \newcommand{\obl}{\textsf{O}}

\newcommand{\ph}{\varphi}
\newcommand{\ps}{\psi}
\newcommand{\ch}{\gamma}

\theoremstyle{definition}
\newtheorem{definition}{Definition}[section]
\newtheorem{theorem}[definition]{Theorem}
\newtheorem{corollary}{Corollary}[definition]
\newtheorem{proposition}[definition]{Proposition}
\newtheorem{lemma}[definition]{Lemma}

\theoremstyle{remark}
\newtheorem{remark}{Remark}

\newcommand{\prop}{\mathsf{Prop}}

\newcommand{\eqthen}{\text{ implies }}

\newcommand{\ent}[2][]{\mathbin{\vdash_{#2}^{#1}}}
\newcommand{\nent}[2][]{\mathbin{\nvdash_{#2}^{#1}}}
\newcommand{\sent}[2][]{\mathbin{\vDash_{#2}^{#1}}}
\newcommand{\nsent}[2][]{\mathbin{\nvDash_{#2}^{#1}}}

\newcommand{\mW}{\mathcal{W}}

\newcommand{\A}{\mathcal{A}}

\newcommand{\mK}{\mathcal{K}}
\newcommand{\mR}{\mathcal{R}}
\newcommand{\mP}{\mathcal{P}}
\newcommand{\mX}{\mathbb{X}}
\newcommand{\mY}{\mathbb{Y}}

\newcommand{\mV}{\mathcal{V}}
\newcommand{\Hi}{\mathcal{H}}
\newcommand{\dP}{\mathbb{P}}
\newcommand{\MP}{{\sf MP}}

\newcommand{\clon}{\mathsf{CLoN}}
\newcommand{\clun}{\mathsf{CLuN}}
\newcommand{\clan}{\mathsf{CLaN}}

\newcommand{\eclon}{\mathsf{CLoNE}}

\newcommand{\lfi}{\textsf{LFI}}
\newcommand{\ldi}{\textsf{LDI}}

\newcommand{\set}[1]{\{ #1 \}}
\newcommand{\tuple}[1]{\langle #1\rangle}

\newcommand{\ptru}[1]{[#1]^{+}}
\newcommand{\mtru}[1]{[#1]^{-}}

\newcommand{\trp}[1]{\llbracket #1 \rrbracket}

\newcommand{\sa}{_{\Sigma}}
\newcommand{\pl}{^{+}}

\newcommand{\mn}{^{-}}

\newcommand{\pls}{_{\Sigma}^{+}}
\newcommand{\mns}{_{\Sigma}^{-}}

\newcommand{\prim}{^\prime}

\title{Modal Extensions for \textsf{CLoN} with Bi-neighborhood Semantics}
\author{Mahan Vaz Silva\\
Ruhr University Bochum\\
Institute for Philosophy I\\ 
mahan.vazsilva@edu.rub.de
\and 
Daniel Skurt\\
Ruhr University Bochum\\
Bochum, Germany\\
Institute for Philosophy I\\
daniel.skurt@rub.de}

\begin{document}
\maketitle

\begin{abstract}
In this paper we will present neighborhood semantics for non-normal modal extensions of $\clon$, which is a sublogic of {\sf FDE}. Our framework is built upon earlier work on {\sf FDE}-based non-normal modal logics and employs two different neighborhood functions for each modal operator. Despite being a logic with a very weak negation operator, we will show that with the right definition of the rejection sets of the modal operators, we can validate non-trivial axioms that contain the weak negation operator. The philosophical aim of our approach is to construct the basis for deontic logics that are able to accommodate both the usual deontic principles and moral dilemmas, without resulting in trivialization of the system. 

\end{abstract}

\section*{Introduction}
In this paper we are going to introduce neighborhood semantics for a subsystem of \textsf{FDE}-based non-normal modal logics, we call $\eclon$, and its extensions. The logic $\eclon$ is a modal expansion of $\clon$, discussed in detail by Diderik Batens et. al. in \cite{Batens1999embedding} as a subsystem of classical logic\footnote{In \cite{Batens80} the system $\clun$ was presented as the system $\sf PI$. In \cite{Batens1999embedding}, however, the name was changed to accommodate the following anagrams: $\clun$ stands for ``\textbf{C}lassical \textbf{L}ogic gl\textbf{u}tty with respect to \textbf{N}egation''; $\clan$, for ``\textbf{C}lassical \textbf{L}ogic g\textbf{a}ppy with respect to \textbf{N}egation''; and $\clon$, for ``\textbf{C}lassical \textbf{L}ogic b\textbf{o}th gappy and glutty with respect to \textbf{N}egation''.}. In its essence, $\clon$ can be interpreted as positive classical logic with a negation operator in its language, but with no axioms nor semantic clauses for negation. Nevertheless, as was shown in \cite{OmoriSkurt2021}, if interpreted in a four-valued \textsf{FDE} like semantics, $\clon$ contains truth-conditions for all operators, including negation, but no falsity conditions, and can be interpreted as a paraconsistent and paracomplete logic. By adding the Law of Excluded Middle to $\clon$, one obtains $\clun$, its paraconsistent subsystem. By adding the Principle of Explosion, one obtains the subsystem $\clan$, which is paracomplete.

Developing neighborhood semantics for modal expansions of logics in the $\clon$ family can be motivated from different perspectives.
First, in \cite{VazMaruchi25} it was defended that, if we were to define a paraconsistent deontic logic to be applied to situations in which conflicting obligations occur, then one should use $\clun$ as the base logic, instead of any of the {\lfi}s\footnote{{\lfi}s refer to the Logics of Formal Inconsistency, developed in \cite{Marcos2007}. The context of the discussion referred to here is the development of the Logics of Deontic Inconsistency (or {\ldi}s) in \cite{Coniglio2009b, Coniglio2009, Peron2008}, which add modal axioms akin to the ones of Standard Deontic Logic to the propositional {\lfi}s.}. And since $\clun$ is an extension of $\clon$, all results obtained for $\clon$ in this paper will carry over to non-normal modal extensions of $\clun$. 
Second, so far as we are aware, there is no extensive study of modalities applied to the $\clon$ family. The original work of Batens, le Clercq and Kurtonina \cite{Batens1999embedding}, only presents $\clon$ as a subsystem of classical logic, while other works, although investigating some modal extensions of the systems based on $\clon$ \cite{Buenosoler2009, ConiglioPeron2013, Coniglio2009b, Coniglio2009}, assume axiom {\sf K} from the start. 
In particular, the authors of \cite{ConiglioPeron2013} propose a system called ${\sf CLoN0.5}$, which assumes the axiom \textsf{K}, a weaker version of necessitation and two extra axioms, one that is a modal version of excluded middle $\Box(\ph \vee \neg \ph)$ and a modal version of explosion $\Box (\ph \to (\neg \ph \to \ps))$. 

Our approach to non-normal modal $\clon$-based logics
assumes nothing but the extensionality rule, rather than specific axioms regarding the modality. This allows us to commit to a minimal modal framework which gives us more flexibility on which axioms to add to our system. To this end we employ the techniques presented in \cite{DalOliNeg18,Drobyshevich2021}, which will allow us to systematically develop a semantics for non-normal modal logics based on $\clon$, as well as $\clun$ and $\clan$, and will naturally connect our results with the results about {\sf FDE}-based non-normal modal logics.

Based on that the paper is structured as follows. In the next section we will give a short overview over $\clon$, $\clun$ and $\clan$, which will serve as our propositional non-modal base logic. In Section \ref{sub_main_weak} we will introduce the non-normal logic $\eclon$ in detail, by describing its Hilbert-calculus and neighborhood semantics, including general proofs of soundness and completeness. This is followed in Section \ref{main_extensions} by presenting various extensions of $\eclon$ containing more than one modal operator and the weak negation operation. Finally, in Section \ref{sec:final} we will briefly discuss the potential application of our framework in the existence of moral dilemmas debate, together with minor technical remarks.


\section{\textsf{CLoN}, \textsf{CLuN}, \textsf{CLaN}}\label{clon}

In this section, we will introduce both Hilbert-calculus and non-deterministic semantics for the base logic $\clon$ and its extensions $\clun$ and $\clan$. To this end, consider a propositional signature $\Sigma$ with one unary connective $\neg$ and  the binary connectives $\wedge$, $\vee$, $\to$. We will furthermore make use of a biconditional $\eq$, defined in the usual way. Let $\mV$ be a denumerable set of propositional variables $\mV=\{p_0,p_1, \ldots\}$ and let $For(\Sigma)$ be the algebra of formulas over $\Sigma$ freely generated by $\mV$. 

\smallskip

\noindent The system $\clon$ is characterized by a Hilbert-calculus $\Hi$ as follows:

\noindent 
{\small
\begin{minipage}{.47\textwidth}
\begin{align*}
& \ph {\to} (\ps {\to} \ph) \tag{{\sf A1}} \label{CpCqp}\\
& (\ph {\to} (\ps {\to} \ch)) {\to} ((\ph {\to} \ps) {\to} (\ph {\to} \ch)) \tag{{\sf A2}} \label{CCpCqrCCpqCpr}\\
& ((\ph {\to} \ps){\to} \ph){\to} \ph \tag{{\sf A3}} \\
& \ph {\to} (\ph {\lor} \ps) \tag{{\sf A4}} \label{CpApq}\\
& \ps {\to} (\ph {\lor} \ps) \tag{{\sf A5}} \label{CqApq}
\end{align*}
\end{minipage}
\ 
\begin{minipage}{.47\textwidth}
\begin{align*}
& (\ph  {\to} \ch ) {\to} ((\ps  {\to} \ch ) {\to} ((\ph \lor \ps ) {\to} \ch )) \tag{{\sf A6}}\\
& (\ph \land \ps) {\to} \ph \tag{{\sf A7}} \label{CKpqp}\\
& (\ph \land \ps) {\to} \ps \tag{{\sf A8}} \label{CKpqq}\\
& (\ch  {\to} \ph ) {\to} ((\ch  {\to} \ps ) {\to} (\ch  {\to} (\ph  {\land} \ps ))) \tag{{\sf A9}} \\
& \displaystyle\frac{\ \ph \quad \ph {\to} \ps\ }{\ps} \tag{{\sf MP}} \label{MP}
\end{align*}
\end{minipage}
}

\begin{remark}
    Note that $\clon$ is essentially positive classical propositional logic in a language with a negation operator that contains no axioms nor rules for $\neg$. However, it can be shown by a Gödel/Dugundji\footnote{See \cite[Th. 11]{Avron2007a}} like construction that it is impossible to construct a finitely many-value deterministic semantics for $\clon$.
\end{remark}

We write $\Gamma \vdash_\clon \ph$ if there is a sequence of formulas $\ps_1, \dots, \ps_n, \ph$, $n\geq 0$, such that every formula in the sequence either (i) belongs to $\Gamma$; (ii) is an axiom of $\Hi$; (iii) is obtained by (\MP) from formulas preceding it in sequence.

 The logics $\clun, \clan$ are obtained from $\clon$ by adding to the axioms all instances of the Principle of Excluded Middle $(\ph \vee \neg \ph)$ and the Principle of Explosion $\ph \to (\neg \ph \to \ps)$, respectively. 

Let us first present a (\emph{non-deterministic}) swap structure semantics for $\clon, \clun, \clan$ without any modal operator. See for instance \cite[Ch. 6]{carnielli2016paraconsistent} and \cite{Coniglio2019a} for more applications of swap structures.

\begin{definition}    
    \label{def:clon-mop}
    A structure $\A= (A, \tilde{\neg}, \tilde{\land}, \tilde{\lor}, \tilde{\to})$ is called a\emph{swap structure} for $\clon$ if $A = \mathbf{2}^2 = \{0,1\}\times\{0,1\}$ and the multioperations are defined in  the following way:\footnote{$\sqcap$, $\sqcup$ and $\supset$ are the operations corresponding to the Boolean meet, join and implication in $\mathbf{2}$, respectively.}
        \[
        \begin{array}{rcl}
			a \tilde{\land} b & := & \{(c_1, c_2) \in A : c_1 = a_1 \sqcap b_1 \}\\
            a \tilde{\lor} b & := & \{(c_1, c_2) \in A : c_1 = a_1 \sqcup b_1 \}\\
            a \tilde{\to} b & := & \{(c_1, c_2) \in A : c_1 = a_1 \supset b_1 \}\\
			\aneg a & := & \{ (c_1, c_2) \in A : c_1 = a_2\}
		\end{array}
        \]
    Let, furthermore, $D = \{(c_1,c_2) \in A : c_1 = 1\}$ be the set of designated values.
\end{definition}

\begin{remark}
Note that the multioperations return sets of values rather than single values, while, as defined in Definition \ref{def:v-clon}, valuations operate on single values. See \cite{AZSurvey} for a general introduction to non-deterministic semantics and \cite{OmoriSkurt2021} for non-deterministic semantics for $\clon$, $\clun$ and $\clan$ in particular.
\end{remark}

\begin{remark}
The four truth values in $A$ could also be expressed in a more standard way as follows: $t = (1,0)$, $b=(1,1)$, $n=(0,0)$, $f=(0,1)$. It is then easy to see that the truth-tables for $\clon$ are generalizations of the truth-tables for \textsf{FDE}, e.g.:
\[
\begin{array}{r|cccc}
    \ph & t & b & n & f \\\hline
    \neg \ph & \{n,f\}& \{t,b\}& \{n,f\}& \{t,b\}
\end{array}
\]
   \end{remark} 

\begin{definition}
    A swap structure for $\clun$ is obtained from the definition above by replacing $A$ with $A^u = \{(x,y) \in A: x \sqcup y = 1  \}$. Similarly a swap structure for $\clan$ is obtained by replacing $A$ with $A^a = \{(x,y) \in A: x \sqcap y = 0  \}$. In other words, $A^u = A\backslash \{(0,0)\}$ and $A^a = A\backslash \{(1,1)\}$.
\end{definition}

\begin{remark}
There are (at least) two corresponding semantics for $\clon$, a two-valued and a four-valued one and the matrix for classical logic is a refinement of the two-valued semantics for $\clon$, see \cite{OmoriSkurt2021}. At the same time it is possible to show that the matrix for \textsf{FDE} is just a refinement of the four-valued Nmatrix for $\clon$, see \cite{AZSurvey} for the notion of refinement. For the purpose of this paper we are employing the four-valued \textsf{FDE}-like setting, since it allows us to address the truth-conditions for negation.
\end{remark}  

\begin{definition}
		\label{def:v-clon}
            Define a function $v: For(\Sigma) \to A$ to be a \textit{swap valuation} for $\clon$ if every condition below is satisfied, for $\ph,\ps\in For(\Sigma)$:
			
			\begin{enumerate}
				\item$v(\ph \land \ps) \in v(\ph) \aland v(\ps)$				
                \item$v(\ph \lor \ps) \in v(\ph) \alor v(\ps)$
                \item $v(\ph \to \ps) \in v(\ph) \ato v(\ps)$
				\item $v(\neg \ph) \in \aneg v(\ph)$								
			\end{enumerate}
	\end{definition}

\begin{remark} \label{remval}
Note that we can also rewrite any swap valuation $v$ as follows $v=(v_1,v_2)$ such that $v_1,v_2: For(\Sigma) \to {\bf2}$. 
Hence, $v(\ph)=(v_1(\ph),v_2(\ph))$ for every formula $\ph$. This means that, for all formulas $\ph$ and $\ps$, where $\sqcap, \sqcup$, and $\supset$ are classical conjunction, disjunction and implication:

\begin{enumerate}\setcounter{enumi}{-1}
\item $v(\ph) \in D$ iff $v_1(\ph)=1$
\item $v_1(\ph \wedge \ps)= v_1(\ph) \sqcap v_1(\ps)$
\item $v_1(\ph \vee \ps)= v_1(\ph) \sqcup v_1(\ps)$
\item $v_1(\ph \to \ps)= v_1(\ph) \supset v_1(\ps)$
\item $v_1(\neg \ph)=v_2(\ph)$
\end{enumerate}
\end{remark}



	\begin{definition}
        \label{def:l-cons}
        Let $\Gamma \cup\{\ph\} \subseteq For(\Sigma)$. We say that $\ph$ is a logical consequence of $\Gamma$ in $\clon$, denoted $\Gamma \vDash_{\clon}  \ph$, iff for all swap valuations $v$, if $v(\ps) \in D$ for all $\ps \in \Gamma$, then $v(\ph) \in D$.\\
        Note that consequence relations for $\clun$ and $\clan$ can be defined accordingly.
	\end{definition}

\begin{remark}
    From the above definition it is easy to see that $\clon$ and $\clun$ are paraconsistent logics, since they do  not validate the Principle of Explosion. $\clan$ on the other validates Explosion. I.e., $\ph, \neg \ph \vdash_\mathsf{CLxN} \ps$, will not be valid for $\mathsf{x} \in \{\mathsf{o},\mathsf{u}\}$ but valid for $\mathsf{x} = \mathsf{a}$. 
\end{remark}

From the above definitions, it is straightforward to prove soundness and completeness for $\clon$, see also \cite{OmoriSkurt2021}.

\begin{theorem}[Soundness and completeness of $\clon$\ w.r.t. swap models] \ \\
    For every $\Gamma \cup \{\ph\} \subset For(\Sigma)$, $\Gamma \vdash_{\clon} \ph$ iff $\Gamma \vDash_\clon  \ph$.
\end{theorem}

\begin{remark}
Note $\ph \eq \ps \nvdash \neg\ph \eq \neg\ps$, hence $\clon$ is not closed under replacement, were $\eq$ is the (weak) equivalence connective defined as usual by the (weak) conditional $\to$. 
\end{remark}

\begin{remark}
Before we close this section, a remark on the status of $\neg$ as a negation operator seems to be in order. While this topic was already discussed to some extent in \cite{OmoriSkurt2021}, we would like to at least mention why $\neg$ could be interpreted as a negation in $\clon$. Typically, the truth conditions for negation are stated as follows: a negation of $\ph$ is true iff $\ph$ is false. Interpreting $\{(c_1,c_2) \in A : c_1 = 1\} = \{t,b\}$ as ``true'' values and $\{(c_1,c_2) \in A : c_2 = 1\} = \{b,f\}$ as ``false'' values, we immediately obtain, that $\neg$ satisfies the truth condition for negation. Skepticism on that status of $\neg$ arises from the fact that in $\clon$ there are no falsity conditions for $\neg$. But, at the very least, ``$\neg$'' can be be considered half a negation.
\end{remark}



Now, for the remainder of this paper let $\Sigma = \langle \neg, \Box, \Diamond, \wedge, \vee, \ra \rangle$. In the coming sections we propose to expand $\clon$ (and as a consequence $\clun$ and $\clan$, as well) by two distinct modal operators $\Box$ and $\Diamond$, that can be interpreted as necessity and possibility and present neighborhood semantics for those expansions. It will be shown, since the negation free fragment of $\clon$ behaves exactly as positive classical logic, that the extensionality rule\[(\ocirc\textsf{Ext})\ \ph \eq \ps /\ocirc\ph \eq \ocirc \ps\]
will be valid, where $\ocirc \in \{\Box,\Diamond\}$. Furthermore, we will recapture positive axioms and rules for $\Box$, and its relative axioms for $\Diamond$ in the usual way and with the usual constraints on neighborhood semantics such that we recover typical axioms such as $\textsf{M,C,N,K,T,4}$. 

We highlight that the property of the negation operator $\neg$ already mentioned above, namely the missing falsity condition, prevents us from defining dual modal operators, as usual. Instead, we inspire ourselves in the spirit of \cite{Dunn1995}, by defining $\Diamond$ independently. The aim is to define an operator that 
is distinct from $\Box$ and not necessarily a \textit{possibility} operator. 
While the symbols $\Box, \Diamond$ might coincide with their usual definition, they do not necessarily behave in such a way. 

\section{The basic non-normal \textsf{CLoN}-based modal logic - \textsf{CLoNE}}\label{sub_main_weak}

\subsection{Hilbert-calculus for \textsf{CLoNE}}

Neighbourhood semantics are typically used for modal logics with an abstract modality $ \ocirc $, where $\ocirc \in \{\Box, \Diamond\}$ satisfying \emph{(strong) extensionality rule}:
\[
(\ocirc\textsf{Ext})\ \ph\eq\ps/\ocirc\ph\eq\ocirc\ps.
\]
For logics with the usual replacement theorem this rule guarantees that weak equivalence still induces a congruence relation on the algebra of formulas. This is not the case, however, for $\clon$-based logics. In this setting, weak negation is not congruential, hence replacement does not hold. 


Nevertheless, we will add this rule and others to $\clon$ ($\clun$, $\clan$). In particular we will also consider the following extensionality principle:
\[
(\neg \ocirc\textsf{Ext})\ \neg\ph\eq\neg\ps/\neg\ocirc\ph\eq\neg\ocirc\ps.
\]





\noindent We denote by $ \eclon$ the smallest logic in a signature over $ \clon$ containing $\Box$ and $\Diamond$, satisfying $(\Box\textsf{Ext})$ and $(\Diamond\textsf{Ext})$.


We will be working with extensions of $\eclon$ by a number of different inference rules and we are identifying axioms with zero-premise rules. By a logic $L$ we will now understand its set of theorems, which is closed under the rules. For a set of rules $\mathcal{R}$ we denote by $\eclon+\mathcal{R}$ the smallest logic closed under all rules in $\mathcal{R}$. 
With any logic $L=\eclon+\mathcal{R}$ we associate the following consequence relation:
\begin{center}
    $\Gamma\ent{L}\ph$ if $\ph$ can be derived from $L\cup\Gamma$ using modus ponens only;
\end{center}
where $\Gamma$ is an arbitrary set of formulas, $\ph$ is a formula. 

\begin{remark}
    Note that the consequence relation, defined that way, is also called \emph{local consequence relation}. Additionally, it is possible to define a \emph{global consequence relation} as follows:  $\Gamma\ent[g]{L}\ph$ if $\ph$ can be derived from $L\cup\Gamma$ using modus ponens, strong extensionality and all rules from $\mathcal{R}$. See \cite{Kracht07} on local and global consequence relations in general. However, since usually global consequence relations are considered to obtain algebraizability results for logics, which is not our aim for this paper, we will only work with the notion of local consequence.
\end{remark}

Several remarks are in order. If $ (r) $ is a zero-premise rule, then closure of $ \Gamma $ under $ (r) $ is simply equivalent to the fact that the conclusion of $ (r) $ belongs to $ \Gamma $. 
Next, note that the consequence relation is defined in such a way that weak extensionality rules (and the other rules in $\mathcal{R}$) are only used to compute $ L $ itself (as a set of formulas) but not directly in derivations. Since modus ponens is the only rule of inference for this consequence relation, it is internalized by weak implication and thus we will be able to obtain the usual deduction theorem with respect to it. 

The following property is then proved routinely.


\begin{theorem}[Deduction Theorem]
Suppose $ L=\eclon+\mathcal{R} $. Then for all $ \Gamma\cup\set{\ph,\ps}\subseteq For(\Sigma)$ we have
\[
\Gamma,\ph\ent{L}\ps\iff\Gamma\ent{L}\ph\ra\ps.
\]
\end{theorem}



\subsection{Neighborhood semantics for \textsf{CLoNE}}

In this section we introduce neighborhood semantics for modal expansions of $\clon$ by two modal operators $\Box, \Diamond$, which we will call $\eclon$. We will employ the framework presented in \cite{Drobyshevich2021} for \textsf{FDE}-based modal logics\footnote{The work above is mentioned for the proximity with the system developed here. The original presentation of bi-neighborhood semantics is given by \cite{DalOliNeg18}.}. To this end, let us first recall some definitions and then introduce the semantic conditions for $\eclon$.



\begin{definition}\label{def_neighborhood}
For a set $W$ we use $\mP W$ to denote the powerset of $W$ and $\dP W$ to denote $\mP W\times\mP W$. We will use $X$, $Y$ to denote elements of $\mP W$ and $\mX$, $\mY$ to denote elements of $\dP W$; for $ \mX\in\dP W $ we will always assume that $\mX=\tuple{X^+,X^-}$. 
An \emph{n-function (neighbourhood function)} on $W$ is a function $N: W\ra\mP\mP W$. 
A bimodal \emph{$\eclon$-frame} is a tuple $\mW=\tuple{W,N\pl_\Box,N\pl_\Diamond }$,
where $W$ is a non-empty set and $N\pl_\Box,N\pl_\Diamond$
are n-functions on $W$. A bimodal \emph{$\eclon$-model} is $\mu=\tuple{\mW,v\pl,v\mn}$, where $\mW$ is a bimodal $\eclon$-frame and $ v\pl, v\mn$ are \emph{neighborhood valuations}, which map elements of $\mV$ to subsets of $W$. For a bimodal $\eclon$-model $\mu$ we will sometimes use $W_\mu$ to denote the underlying set of worlds, but if it is clear from the context we will omit it.  We will also relativize swap valuation $v(\ph)=(v_1(\ph),v_2(\ph))$ to elements of $W$, as follows: $v_w(\ph)=(v_{(1,w)}(\ph),v_{(2,w)}(\ph))$.\end{definition}

\begin{definition}\label{def_verif}
For any formula $\ph$, let $\ptru{\ph}_{\mu} \subseteq W$ and $\mtru{\ph}_{\mu} \subseteq W$ be its \emph{verification set} and \emph{rejection set}, respectively, such that the following conditions are satisfied:
\begin{center}
	\begin{tabular}{l l }
		$\ptru{p}_{\mu}$ & $= v\pl(p) =\{w \in W : v_{(1,w)}(p)=1\} $\\ 
		$\mtru{p}_{\mu}$ & $= v\mn(p) =\{w \in W : v_{(2,w)}(p)=1\} $\\
		$\ptru{\neg\ph}_{\mu}$ & $=\mtru{\ph}_{\mu}$\\
		$\ptru{\ph\wedge\ps}_{\mu}$ & $=\ptru{\ph}_{\mu}\cap\ptru{\ps}_{\mu}$; \\
		$\ptru{\ph\vee\ps}_{\mu}$ & $=\ptru{\ph}_{\mu}\cup\ptru{\ps}_{\mu}$; \\
		$\ptru{\ph\ra\ps}_{\mu}$ & $=\overline{\ptru{\ph}_{\mu}}\cup\ptru{\ps}_{\mu}$; \\
		$\ptru{\Box\ph}_{\mu}$ & $=\set{w\in W\mid \ptru{\ph}_\mu\in N\pl_\Box(w)}$; \\
        $\ptru{\Diamond\ph}_{\mu}$ & $=\set{w\in W\mid\ptru{\ph}_\mu\in N\pl_\Diamond(w)}$; \\
	\end{tabular}
\end{center}

\noindent where $\overline{X}:=W\setminus X$. For a set of formulas $\Gamma$ put $\ptru{\Gamma}_{\mu}:=\bigcap_{\ph\in\Gamma}\ptru{\ph}_{\mu}$, with $\ptru{\varnothing}_{\mu}:=W_\mu$. Furthermore, for convenience and later use we define $\trp{\ph}_\mu:=\tuple{\ptru{\ph}_{\mu},\mtru{\ph}_{\mu}}$ as the \emph{proof pair of $\ph$}. In \cite{MaLin20c} $\ptru{\ph}_{\mu}$,  $\mtru{\ph}_{\mu}$ and $ \trp{\ph}_\mu$ are called the \emph{polarity}. 

Bimodal $\mathsf{CLuNE}$- and $\mathsf{CLaNE}$-frames can be obtained by forcing $\ptru{\ph}\cup \mtru{\ph} = W$ and $\ptru{\ph}\cap \mtru{\ph} = \emptyset$, respectively, for any $\ph$.
\end{definition}

\begin{remark}
Note that in the above definition some formulas $\ph$ will have arbitrary rejection sets, because we do not impose any conditions on the rejection sets for $\neg$, $\wedge$, $\vee$ and $\to$. But this way we will stay faithful to the interpretation of the operators in $\clon$. In particular, since the rejection set are arbitrary, typical principles like the De Morgan laws or double negation elimination or introduction are not valid in $\clon$. Nevertheless, for a concrete bimodal \emph{$\eclon$-model} $\mu$, verification and rejection sets are always defined. In fact, every model can be equally expressed by the set of all of its proof pairs.
\end{remark}

\begin{remark}
Furthermore, we can use above definitions to define corresponding operations on $\dP W$ in a way, which is very reminiscent of how the operations are defined over swap-structures for $\clon$. For instance, for $\mX,\mY\in\dP W$ put:
\[
\sneg\mX:=\tuple{X^-,Z};\quad
\mX\wedge\mY:=\tuple{X^+\cap Y^+,Z};\quad
\mX\vee\mY:=\tuple{ X^+\cup Y^+, Z}.
\] 
Note $Z \in \mP W$ is arbitrary.
\end{remark}





\begin{proposition}\label{p-trueset-prop}
	Suppose $\Gamma$ and $\Delta$ are sets of formulas, then:
	\begin{enumerate}
		\item $\Gamma\subseteq\Delta\eqthen\ptru{\Delta}_{\mu}\subseteq\ptru{\Gamma}_{\mu}$;
		\item $\ptru{\Gamma\cup\Delta}_{\mu}=\ptru{\Gamma}_{\mu}\cap\ptru{\Delta}_{\mu}$;
		\item if $\Delta$ is the closure of $\Gamma$ under modus ponens then $\ptru{\Gamma}_{\mu}=\ptru{\Delta}_{\mu}$.
	\end{enumerate}
\end{proposition}

\begin{definition}
For any class $\mK$ of $\eclon$-frames we define its corresponding logic as follows:
\[
L(\mK):=\bigcap_{\mW\in\mK}\set{\ph\mid\ptru{\ph}_{\mu}=W_\mu\text{ for any }\eclon\text{-model }\mu\text{ over }\mW}.
\]
We also define the semantic counterpart of the consequence relation as follows:
\begin{center}
    $\Gamma\sent{\mK}\ph$ if $\ptru{\Gamma}_{\mu}\subseteq \ptru{\ph}_{\mu}$ for every $\mW\in\mK$ and $\eclon$-model $\mu$ over $\mW$;
\end{center}   
\end{definition}

\noindent Observe that clearly for any formula $\ph$ we have $\ph\in L(\mK)\iff\varnothing\sent{\mK}\ph$

\begin{definition}
 Let $\mK$ be a class of frames.
We say that a rule $\ph_1,\dotsc,\ph_n/\ps$ is \emph{valid} in $\mK$ iff, if $\ptru{\ph_1}=\dotsc=\ptru{\ph_n} = W_\mu$ then $\ptru{\ps} = W_\mu$. Then our definitions imply that a zero-premise rule $\varnothing/\ph$ is valid in $\mK$ iff $\ph\in L(\mK)$.  A set of rules $\mathcal{R}$ \emph{is valid in} $\mK$ if all rules in $\mathcal{R}$ are valid in $\mK$. A logic $L=\eclon+\mR$ \emph{is valid in} $\mK$ if $\eclon\subseteq L(\mK)$ and all rules in $ \mathcal{R} $ are valid in $\mK$; in this case $\mW\in\mK$ is called an \emph{$L$-frame}.
An extension $L$ of $\eclon$ is 
\emph{sound and complete} with respect to a class $\mK$ of $\eclon$-frames if for any set of formulas $\Gamma\cup\set{\ph}$ we have
\[
\Gamma\ent{L}\ph\iff\Gamma\sent{\mK}\ph
\]   
\end{definition}

The following simple proposition will be helpful:

\begin{proposition}\label{p_aux}
	Suppose $\mu$ is a $\eclon$-model and $\ph$, $\ps$ are two formulas, then
\begin{enumerate}
	\item $\ptru{\ph\ra\ps}_{\mu}=W_\mu\iff\ptru{\ph}_\mu\subseteq\ptru{\ps}_\mu$;
	\item $\ptru{\ph\eq\ps}_{\mu}=W_\mu\iff\ptru{\ph}_\mu=\ptru{\ps}_\mu \iff \ptru{\Box\ph}_\mu=\ptru{\Box\ps}_\mu \iff \ptru{\Diamond\ph}_\mu=\ptru{\Diamond\ps}_\mu$;
\end{enumerate}
\end{proposition}


\begin{remark}
It is easy to see that some typical rules containing the negation operator, will not be valid in $\mK$, for obvious reasons. For example, $\neg \ph \leftrightarrow \neg \ps /  \Box \ph \leftrightarrow  \Box \ps$, since $\mtru{\ph} = \mtru{\ps}_{\mu} \nRightarrow \ptru{ \Box\ph}_\mu=\ptru{\Box\ps}_\mu$. 



     
 
\end{remark}

\begin{remark}
    We note that $\eclon$ and its extensions are deontically paraconsistent \cite{Coniglio2009b}, i.e., there are formulas $\ph, \ps$  such that $\obl \ph, \obl \neg \ph \nvdash_{\mathsf{CLxN}} \obl \ps)$, for $x \in \set{\sf o,u,a}$. To see that it fails, suppose $\ptru{\obl \ph \to (\obl \neg \ph \to \obl \ps)} = W_\mu$. That means that $\ptru{\obl \ph} \subseteq \ptru{\obl \neg \ph \to \obl \ps}$, by Proposition \ref{p_aux}. Since $\ptru{\obl \neg \ph \to \obl \ps} = \overline{\ptru{\obl \neg \ph}} \cup \ptru{\obl \ps}$, then $\ptru{\obl \ph} \subseteq \overline{\ptru{\obl \neg \ph}} \cup \ptru{\obl \ps}$. However, in our construction the latter is not valid. 
\end{remark}

\subsection{Soundness}\label{sub_main_sound}


\begin{theorem}
    Suppose $\mK$ is an arbitrary class of $\eclon$-frames then all axioms of $\clon$ are valid in $\mK$. 
\end{theorem}

\begin{proof}

    We only show the case for {\sf A3}, as the other axioms can be proven in a similar manner. Assume for contradiction $\ptru{((\ph {\to} \ps){\to} \ph){\to} \ph}_\mu \not= W_\mu$,  by Definition \ref{def_verif} and Proposition \ref{p_aux} item 1., we have the following:  
   
    \begin{align*}
        \ptru{(\ph \to \ps)\to \ph}_\mu \nsubseteq \ptru{\ph}_\mu & \text{ iff } \overline{\ptru{\ph \to \ps}_\mu} \cup \ptru{\ph}_\mu \nsubseteq \ptru{\ph}_\mu  \\
         & \textsf{ iff } (\overline{\overline{\ptru{\ph}_\mu}} \cap \overline{\ptru{\ps}_\mu}) \cup \ptru{\ph}_\mu \nsubseteq \ptru{\ph}_\mu\\
        & \textsf{ iff } (\ptru{\ph}_\mu \cap \overline{\ptru{\ps}_\mu}) \cup \ptru{\ph}_\mu \nsubseteq \ptru{\ph}_\mu\\
        & \textsf{ iff } \ptru{\ph}_\mu \nsubseteq \ptru{\ph}_\mu
    \end{align*}
   
    \noindent which is impossible.    
\end{proof}

\begin{theorem}[Soundness]\label{t-basic-soundness}
	Suppose $\mK$ is an arbitrary class of $\eclon$-frames then the logic $\eclon$ and rules modus ponens, $\Box$-extensionality and $\Diamond$-extensionality are valid in $\mK$.
\end{theorem}

\begin{proof}
    We check first that the rules are valid in $\mK$. Let $\Gamma \subseteq For(\Sigma)$ such that $\ptru{\Gamma}_\mu = W_\mu$ and $\ph, \ph\ra \ps \in \Gamma$. Now clearly this means $w \in \ptru{\ph}_\mu$ and $w \in \ptru{\ph \ra \ps}_\mu$ for all $w \in W_\mu$, and hence by Proposition \ref{p_aux} $w \in \ptru{\ps}_\mu$ for all $w \in W_\mu$. For extensionality, it is easy to observe that again by Proposition \ref{p_aux} we have $\ptru{\ph \eq \ps}_\mu = W_\mu$ iff $\ptru{\ph}_\mu=\ptru{\ps}_\mu$, from which it follows that $\ptru{\Box\ph}_\mu = \ptru{\Box\ps}_\mu$. An analogous argument shows the validity of $\Diamond\textsf{Ext}$. 
     
    To see that $\eclon \subseteq L(\mK)$ notice that every $\ph$ in $\eclon$ is either an instance of an axiom, or is obtained from the axioms by $\MP$, $\Box\textsf{Ext}$ or $\Diamond\textsf{Ext}$ and since these rules are valid in every $\eclon$ model for every frame, then $\ph \in L(\mK)$.
\end{proof}

The following result shows that validity of a logic in a class implies both (local) soundness. 

\begin{theorem}\label{t-ext-soundness}
Suppose $L=\eclon+\mathcal{R}$. If $\mK$ is some class of $L$-frames, then for any set $\Gamma\cup\set{\ph}$ of formulas we have:
\[
(\Gamma\ent{L}\ph\Longrightarrow\Gamma\sent{\mK}\ph).
\]
\end{theorem}
\begin{proof}

We consider the model $\mu = \tuple{\mW, v\pl, v\mn}$, where $\mW \in \mK$, thus $\mW$ is an $L$-frame.
To show the result, we assume that $\Gamma \ent{L}\ph$. This means that $\ph$ is derived from $\Gamma \cup L$ using only modus ponens. Let $\Delta$ be the closure of $\Gamma \cup L$ under modus ponens. Then, by Proposition \ref{p-trueset-prop}, $\ptru{\Gamma \cup L}_\mu = \ptru{\Gamma}_\mu \cap \ptru{L}_\mu = \ptru{\Delta}_\mu$. This also implies that $\set{\ph} \subseteq \Delta$, hence $\ptru{\Delta}_\mu \subseteq \ptru{\ph}_\mu$. We use the previous theorem to assert that $L \subseteq T_\mu$, where $T_\mu:= \set{\ph \mid\ \ptru{\ph}_\mu=W_\mu}$. From this we have $\ptru{L}_\mu = W_\mu$, since $T_\mu$ is closed under modus ponens and any rule in $\mR$, while also retaining $\eclon \subseteq T_\mu$ and $\ptru{T_\mu}_\mu = W_\mu$. Finally, the result follows from $\ptru{\Gamma}_\mu = \ptru{\Gamma}_\mu \cap W_\mu = \ptru{\Gamma \cup L}_\mu$ . \qedhere

\end{proof}

\subsection{Completeness}\label{sub_main_complete}

We are now ready to prove some basic completeness results. For the remainder of this section let us fix an extension $L=\eclon+\mathcal{R}$. We say that a set of formulas $ \Gamma $ is
\begin{enumerate}
	\item \emph{non-trivial}, if it does not coincide with the set of all formulas;
	\item a \emph{(local) $ L $-theory}, if $\Gamma\ent{L}\ph$ implies $\ph\in\Gamma$ for any formula $ \ph $;
	\item satisfies the \emph{disjunction property}, if $\ph\vee\ps\in\Gamma$ implies $\ph\in\Gamma$ or $\ps\in\Gamma$ for any $ \ph,\ps $;
	\item is a \emph{prime $ L $-theory}, if it is a non-trivial $ L $-theory, which satisfies the disjunction property.
\end{enumerate}


\noindent As usual, we have

\begin{lemma}[Extension]\label{lem_compl_extension}
	If $\Gamma\nent{L}\ph$ then there is a prime $L$-theory $\Delta$ such that $\Gamma\subseteq\Delta$ and $\Delta\nent{L}{\ph}$.
\end{lemma}



Suppose $ \Sigma $ is a non-trivial 
$ L $-theory. By $ W\sa $ denote the set of all prime $ L $-theories containing $ \Sigma $. Clearly, $ W\sa\neq\varnothing $ by the Extension lemma. If, additionally, $\ph$ is some formula, put:
\[
\ptru{\ph}\sa:=\set{\Gamma\in W\sa\mid\ph\in\Gamma};\quad
\mtru{\ph}\sa:=\set{\Gamma\in W\sa\mid\neg\ph\in\Gamma}
\]
As before, for a set of formulas $\Gamma$ put $\ptru{\Gamma}\sa:=\bigcap_{\ph\in\Gamma}\ptru{\ph}\sa$ with $\ptru{\varnothing}\sa:=W\sa$. It is easy to see that $\ptru{\Gamma}\sa=\set{\Delta\in W\sa\mid\Gamma\subseteq\Delta}$.
We say that the n-function $N\pl_\Box$, $N\pl_\Diamond$ on $ W\sa $ are \emph{canonical (with respect to $\Sigma$)} if for any $\Gamma\in W\sa$ and formula $\ph$ we have:
\[\ptru{\Box\ph}_\Sigma=\set{\Gamma\in W\sa\mid \ptru{\ph}\sa\in N\pl_\Box(\Gamma)} \qquad
\ptru{\Diamond\ph}_\Sigma=\set{\Gamma\in W\sa\mid \ptru{\ph}\sa\in N\pl_\Diamond(\Gamma)}\]

\bigskip


A $\eclon$-frame $\mW\sa=\tuple{W\sa,N\pl_\Box,N\pl_\Diamond}$ is a \emph{canonical $L$-frame (with respect to $\Sigma$)} if both $N\pl_\Box$ and $N\pl_\Diamond$ are canonical with respect to $\Sigma$.

For a non-trivial 
$ L $-theory $ \Sigma $ and a canonical $ L $-frame $ \mW\sa $ we say that $\mu\sa=\tuple{\mW\sa,v\pl,v\mn}$ is a \emph{canonical $L$-model (with respect to $\Sigma$)} if for any $p\in\prop$ we have
\[
v\pl(p)=[p]\pls;\qquad
v\mn(p)=[p]\mns.
\]

\begin{lemma}[Canonical]\label{lem_compl_canonical}
	Suppose $\mu\sa=\tuple{W\sa,N\pl_\Box,N\pl_\Diamond,v\pl,v\mn}$ is a canonical $L$-model with respect to a non-trivial 
    $L$-theory $\Sigma$. Then for any $\Gamma\in W\sa$ and formula $\ph$ we have
	\[
	\ptru{\ph}_{\mu\sa}=\ptru{\ph}\sa;\qquad
	\mtru{\ph}_{\mu\sa}=\mtru{\ph}\sa.
	\]
\end{lemma}
\begin{proof}
	As usual, by induction on the complexity of $\ph$. Let $\Gamma\in W\sa$, then we can distinguish the following cases (the cases for $\wedge$, $\vee$ and $\ra$ are proved in a similar manner):
    \begin{itemize}
    \item $\ph=\neg\ps$:\\
	$\Gamma\in\ptru{\neg\ps}_{\mu\sa}\iff\Gamma\in\mtru{\ps}_{\mu\sa}\iff
	\Gamma\in\mtru{\ps}\sa\iff\Gamma\in\ptru{\neg\ps}\sa.$
        
    \item $\ph=\ocirc\ps$ for $\ocirc \in \{\Box, \Diamond\}$:\\
	$\Gamma\in\ptru{\ocirc\ps}_{\mu\sa}\iff\ptru{\ps}_{\mu\sa}\in N\pl_\ocirc(\Gamma)\iff
	\ptru{\ps}\sa\in N\pl_\ocirc(\Gamma)\iff\Gamma\in\ptru{\ocirc\ps}\sa.$ \qedhere
    \end{itemize}
\end{proof}

\begin{theorem}[Completeness]\label{t-l-com-basic}
Let $L$ be a canonical extension of $\eclon$. Then $L$ is complete with respect to the class $\mK$ of all $L$-frames.
\end{theorem}
\begin{proof}
	Suppose $\Gamma\nent{L}\ph$, then by the Extension lemma there is a prime $L$-theory $\Gamma'$ such that $\Gamma\subseteq\Gamma'$ and $\Gamma'\nent{L}\ph$.
	Put 
    $\Sigma=L$, then $\Sigma$ is non-trivial.
	Consider some canonical $L$-model $\mu\sa=\tuple{\mW\sa,v\pl,v\mn}$ with respect to $\Sigma$, such that $L$ is valid in $\mW\sa$. 
    Then we have $\Gamma'\in\ptru{\Gamma}\sa$ and $\Gamma'\notin\ptru{\ph}\sa$. Then by the Canonical lemma we obtain $\ptru{\Gamma}_{\mu\sa}=\ptru{\Gamma}\sa\not\subseteq\ptru{\ph}\sa=\ptru{\ph}_{\mu\sa}$.
Then we have $\Gamma\nsent{\mW\sa}\ph$ and since $\mW\sa\in\mK$, we also have $\Gamma\nsent{\mK}\ph$.
\end{proof}


\begin{corollary}
$\eclon$ is sound and complete with respect to the class of all $\eclon$-frames. 
\end{corollary}


\section{Extensions of \textsf{CLoNE}}\label{main_extensions}

In this section we will consider various axiomatic extensions of $\eclon$. 
Since in $\eclon$ semantics all falsity conditions (rejection sets) for complex formulas are omitted, it can be seen as too weak to be interesting by itself. Despite of that, as shown in Section \ref{sub_extensions}, positive extensions of $\eclon$ behave as expected. However, in order to be able to express modal axioms with negation and some well-known interaction axioms, while staying faithful to the base logic $\clon$, 
we will extend our semantics given in Definition \ref{def_verif} with two new conditions describing the rejection sets for $\Box$ and $\Diamond$:

\begin{center}
\begin{tabular}{l l}
$\mtru{\Box\ph}_{\mu}$ & $=\set{w\in W\mid \mtru{\ph}_\mu\in N\mn_\Box(w)}$; \\
$\mtru{\Diamond\ph}_{\mu}$ & $=\set{w\in W\mid\mtru{\ph}_\mu\in N\mn_\Diamond(w)}$; \\
\end{tabular}    
\end{center}

\noindent where $N\mn_\Box$ and $N\mn_\Diamond$ are two n-functions, independent of $N\pl_\Box$ and $N\pl_\Diamond$. We therefore consider now bimodal \emph{$\eclon$-frames} to be tuples of the form $\mW=\tuple{W,N\pl_\Box,N\pl_\Diamond,N\mn_\Box,N\mn_\Diamond }$, where $W$ is a non-empty set and $N\pl_\Box,N\pl_\Diamond,N\mn_\Box,N\mn_\Diamond$ are n-functions on $W$. This move allows us to evaluate negated modalities such as $\neg \Box$ or $\Diamond \neg$. 
Then all of the definitions and theorems from Section \ref{sub_main_weak} carry over for the extended semantics with only some slight alterations. 

Additionally, because of Proposition \ref{p_aux}, $\neg \ocirc\textsf{Ext}$ (where $\ocirc \in \{\Box, \Diamond\}$) will be valid, since $\ptru{\neg\ph\eq\neg\ps}_{\mu}=W_\mu\iff\ptru{\neg \ph}_\mu=\ptru{\neg\ps}_\mu \iff \mtru{\ph}_\mu=\mtru{\ps}_\mu \iff \mtru{\ocirc\ph}_\mu=\mtru{\ocirc\ps}_\mu \iff \ptru{\neg\ocirc\ph}_\mu=\ptru{\neg\ocirc\ps}_\mu \iff \ptru{\neg\ocirc\ph\eq\neg\ocirc\ps}_{\mu}=W_\mu$. We therefore have to extend the Hilbert-calculus for $\eclon$ by
\[
(\neg \ocirc\textsf{Ext})\ \neg\ph\eq\neg\ps/\neg\ocirc\ph\eq\neg\ocirc\ps.
\]

\noindent Soundness is then proved as above, by considering the extended Hilbert-calculus and completeness requires only a small alteration of the Extension Lemma \ref{lem_compl_extension}: 
\begin{quote}
We say that the n-functions $N\mn_\ocirc$ are \emph{canonical (with respect to $\Sigma$)} if for any $\Gamma\in W\sa$ and formula $\ph$ we have:
$\mtru{\ocirc\ph}_\Sigma=\set{\Gamma\in W\sa\mid \mtru{\ph}\sa\in N\mn_\ocirc(\Gamma)}$ and a $\eclon$-frame $\mW\sa=\tuple{W\sa,N\pl_\Box,N\pl_\Diamond,N\mn_\Box,N\mn_\Diamond}$ is a \emph{canonical $L$-frame (with respect to $\Sigma$)} if all $N\pl_\ocirc$ and $N\mn_\ocirc$ are canonical with respect to $\Sigma$.    
\end{quote}

\noindent For the extended proof of Lemma \ref{lem_compl_canonical} we only require the additional cases for the rejections of the modal operators, then completeness follows: 
\begin{quote}
$\Gamma\in\mtru{\ocirc\ps}_{\mu\sa}\iff\mtru{\ps}_{\mu\sa}\in N\mn_\ocirc(\Gamma)\iff
	\mtru{\ps}\sa\in N\mn_\ocirc(\Gamma)\iff\Gamma\in\mtru{\ocirc\ps}\sa.$   
\end{quote}
	


\subsection{Monomodal Extensions}\label{sub_extensions}

We start by showing that positive monomodal\footnote{Under monomodal extensions of $\eclon$ we understand extensions by axioms that contain only one (or multiple instances of the same) modal operator.} extensions of $\eclon$, i.e. those not requiring negation, behave in that framework as expected and then briefly discuss monomodal extensions with negation. We leave axioms that have two different modalities for Section \ref{sub_int_extensions}. In order to simplify our presentation, we consider extensions for both modal operators, $\Box$ and $\Diamond$, simultaneously. Let $\ocirc \in \{\Box, \Diamond\}$ and consider the following axioms:

\begin{center}
\begin{tabular}{ll}
    $\sf M$ & $\ocirc (\ph \land \ps) \to (\ocirc \ph \land \ocirc \ps)$\\
    $\sf C$ & $ \ocirc \ph \land \ocirc \ps \to \ocirc (\ph \land \ps)$\\
    $\sf K$ & $\ocirc (\ph \to \ps) \to (\ocirc \ph \to \ocirc \ps)$\\
    $\sf M_\lor$ & $\ocirc (\ph \lor \ps) \to \ocirc \ph \lor \ocirc \ps$\\
    $\sf T$ & $\ocirc \ph \to \ph$\\
    $\sf 4$ & $\ocirc \ph \to \ocirc \ocirc \ph$
\end{tabular}
\end{center}

\noindent Then we easily obtain the following lemma.

\begin{lemma}\label{lem_posext}
A $\eclon$-frame which satisfies for any proof pairs $\mX$ and $\mY$
\begin{enumerate}
    \item if $X\pl \cap Y\pl \in N\pl_\ocirc(w)$, then $X\pl,Y\pl \in N\pl_\ocirc(w)$ (\emph{Augmentation}) corresponds to $\sf M$,
    \item if $X\pl,Y\pl \in N\pl_\ocirc(w)$, then $X\pl\cap Y\pl \in N\pl_\ocirc(w)$ (\emph{Closure by Intersection}) corresponds to $\sf C$,
    \item both \emph{Augmentation} and \emph{Closure by Intersection} corresponds to $\sf K$,
    \item if $X\pl \cup Y\pl \in N\pl_\ocirc(w)$, then $X\pl \in N\pl_\ocirc(w)$ or $Y\pl \in N\pl_\ocirc(w)$ (\emph{Primeness}) corresponds to $\sf M_\lor$,
    \item if $X\pl \in N\pl_\ocirc(w)$, then $w \in X\pl$ (\emph{Reflexivity}) corresponds to $\sf T$,
    \item if $X\pl \in N\pl_\ocirc(w)$, then $\set{w\prim \in W : X\pl \in N\pl_\ocirc(w\prim)} \in N\pl_\ocirc(w)$ (\emph{Tran\-si\-ti\-vi\-ty}) corresponds to $\sf 4$.
\end{enumerate}
\end{lemma}

\begin{proof}
See appendix.    
\end{proof}

\begin{remark}
    We note that the properties above are similar to their standard versions. This is a result of the positive fragment of $\eclon$ being equivalent to the positive fragment of classical logic. In that sense, the restrictions to neighborhood semantics that lead to certain classes of frames, for example, neighborhood frames satisfying Augmentation corresponding to logics that satisfy $\textsf{M}$. 
\end{remark}

\noindent Let us now consider the following axioms:
\begin{center}
\begin{tabular}{ll}
    $\sf M\mn$ & $\ocirc (\neg \ph \land \neg \ps) \to (\neg \ocirc \ph \land \neg \ocirc \ps)$\\
    $\sf C\mn$ & $ \neg \ocirc \ph \land \neg \ocirc \ps \to \ocirc (\neg \ph \land \neg \ps)$\\
    $\sf M\mn_\lor$ & $\ocirc (\neg \ph \lor \neg \ps) \to \neg \ocirc \ph \lor \neg \ocirc \ps$\\
    $\sf T\mn$ & $\neg \ocirc \ph \to \neg \ph$\\
\end{tabular}
\end{center}

\noindent Then we easily obtain the following lemma.
\begin{lemma}\label{lem_negext} A $\eclon$-frame which satisfies for any proof pairs $\mX$ and $\mY$
\begin{enumerate}
    \item if $X\mn\cap Y\mn \in N\pl_\ocirc(w)$ then $X\mn,Y\mn \in N\mn_\ocirc(w)$ corresponds to $\sf M\mn$, 

    
    \item if $X\mn,Y\mn \in N\mn_\ocirc(w)$, then $X\mn\cap Y\mn \in N\pl_\ocirc(w)$ corresponds to $\sf C\mn$, 


    \item if $X\mn \cup Y\mn \in N\pl_\ocirc(w)$, then $X\mn \in N\mn_\ocirc(w)$ or $Y \in N\mn_\ocirc(w)$ corresponds to $\sf M\mn_\lor$,


    \item if $X\mn \in N\mn_\ocirc(w)$, then $w \in X\mn$ corresponds to $\sf T\mn$.

\end{enumerate}
\end{lemma}
\begin{proof}
See appendix.    
\end{proof}

\begin{remark}\label{rem_undefinable}
It is worth noting that because of the weak negation operator formulas that contain $\neg \ocirc \neg \ph$, e.g. $\ocirc \ph \to \neg \ocirc \neg \ph$, can not be captured by conditions of the frame alone, since $w\in \ptru{\neg \ocirc \neg \ph}_\mu = \mtru{\ocirc \neg \ph}_\mu$ and therefore $\mtru{\neg \ph}_\mu \in N\mn_\ocirc(w)$. But the rejection set of $\neg \ph$ is not defined and can therefore not be described by corresponding frame conditions. 
\end{remark}

\begin{remark}
Note that the lists of axioms considered in this section are of course not exhaustive but rather stand as examples of the possibilities this semantics offers. 
\end{remark}

\subsection{Multimodal Extensions}\label{sub_int_extensions}

As in Section \ref{sub_extensions} we want to simplify our presentation. Hence, we consider extensions for both modal operators, $\Box$ and $\Diamond$, simultaneously and let $\ocirc,\bcirc \in \{\Box, \Diamond\}$. Note that $\ocirc$ and $\bcirc$ do not need to be different. Now, consider the following axioms:

\begin{center}
\begin{tabular}{ll}
    $\sf D$ & $\ocirc \ph \to \bcirc \ph$\\
    $\sf D\mn$ & $\neg \ocirc \ph \to \neg \bcirc \ph$\\
    $\sf I_1$ & $\neg \ocirc \ph \to \bcirc \neg \ph$\\
    $\sf I_2$ & $\ocirc \neg \ph \to \neg \bcirc \ph$\\
\end{tabular}
\end{center}

\noindent Then we easily obtain the following lemma.
\begin{lemma}\label{lem_intext} A $\eclon$-frame which satisfies for any proof pairs $\mX$ and $\mY$
\begin{enumerate}
    \item if $X\pl \in N\pl_\ocirc(w)$, then $X\pl \in N\pl_\bcirc(w)$ corresponds to $\sf D$,


    \item if $X\mn \in N\mn_\ocirc(w)$, then $X\mn \in N\mn_\bcirc(w)$ corresponds to $\sf D\mn$,


    \item if $X\mn \in N\mn_\ocirc(w)$, then $X\mn \in N\pl_\bcirc(w)$ corresponds to $\sf I_1$,


    \item if $X\mn \in N\pl_\ocirc(w)$, then $X\mn \in N\mn_\bcirc(w)$ corresponds to $\sf I_2$.

\end{enumerate}
\end{lemma}

\begin{proof}
See appendix.    
\end{proof}

\begin{remark}
We just want to mention that it is also entirely possible to add even more axioms showing the interaction between the two different modalities. For example:
\begin{itemize}
    \item $\ocirc (\ph \land \ps) \to (\bcirc \ph \land \bcirc \ps)$ corresponds to if $X\pl \cap Y\pl \in N\pl_\ocirc(w)$, then $X\pl,Y\pl \in N\pl_\bcirc(w)$,

    \item $\ocirc (\ph \land \ps) \to (\neg\bcirc \ph \land \neg \bcirc \ps)$ corresponds to if $X\pl \cap Y\pl \in N\pl_\ocirc(w)$, then $X\mn,Y\mn \in N\mn_\bcirc(w)$.
\end{itemize}
But due to space restrictions, we will leave a thorough investigation of all possibilities for another time.
\end{remark}

\section{Philosophical remarks}\label{sec:philosophybit}

Although the main objective of the paper is to present a possibility of constructing modal versions of $\clon$, we believe that the system can also be put to good use in philosophical practice, although, due to its weak 
negation, its applicability might be rather restricted. The main argument for the use of such logics lie, for example, in \cite{VazMaruchi25}, where the authors get to the conclusion that, if paraconsistent logic is used as the basis for deontic purposes in the context of the $\lfi$s, then the base logic used should be $\clun$, an extension of $\clon$. 

However, as presented in \cite{Coniglio2009b, Coniglio2009, Peron2008, VazConiglio2025}, the semantics of paraconsistent deontic logics are presented using Kripke frames. This means that at least axiom $\sf K$ holds. Following \cite{Marcus1980} and others, Aggregation, corresponding to property $\sf C$ in Section \ref{sub_extensions}, is not a desired property for deontic systems. This fact motivates the presentation of the semantics in terms of neighborhood semantics, since we can avoid this commitment to stronger axioms. 

Another advantage to use ${\clon}$ as a base logic is the fact that $\clon$-negation cannot express modal notions as duals of each other. It requires that our frames keep track of each of the modal notions independently of each other, and thus, if we want to express $n$-many modal notions, we need $n$-many distinct neighborhood functions. Although prima facie this is not desirable, it also allows us to create multimodal systems in which each of the modal concepts are independent of each other, but that they can come together through bridge principles, as for example presented in Section \ref{sub_int_extensions}. Again, one example for the use of such a system, as discussed in \cite{VazMaruchi25}, would be the argument in favor of the existence of moral dilemmas while preserving most of the intuitive moral principles.

In Ethics, moral dilemmas are usually formalized through a bimodal logic, where $\obl$ denotes obligation and $\Diamond$ denotes metaphysical possibility, as follows \[\obl \ph \land \obl \psi \land \neg \Diamond (\ph \land \psi) \]

We call \textit{genuine moral dilemmas} those in which the obligations involved in a dilemma are not defeasible, i.e., they are obligations \textit{all-things-considered}. It is important to highlight this fact since referees always point out that there are similar alternatives to conflicts of obligations which are defeasible or involve non-monotonic reasoning of some sort. Our aim is also to provide a formal tool to cases in which defeasibllity cannot help us.

The problem with accepting the existence of genuine moral dilemmas is that they lead to contradictions when accepting the following pairs of principles, frequently taken to be intuitive and generally accepted as moral principles\footnote{A long debate about the existence of moral dilemmas can be found in a number of papers, for example \cite{Bohse2005,Conee1982, Holbo2002, Marcus1980, McConnell1978, ToddWeber2002, Weber2007}. The position that  defends moral dilemmas do not exist, take moral principles to be accepted based on their intuitiveness. Given they are on a par with our usually accepted logical systems, rejecting them would go against our rationality. Expressing moral dilemmas formally then lead to contradictions, thus, to trivialization, hence, their existence cannot be accepted together with the deontic principles. In response to this position, some authors defend a dialetheist account of moral systems, in order to accept moral dilemmas \cite{RibeiroTeles2021,  VazMaruchi25, Weber2007}.}:

\renewcommand{\arraystretch}{1.2}
\begin{center}
\begin{tabular}{c|c}
  Pair 1   &  Pair 2\\\hline
$\obl \ph \to \Diamond \psi$ & $\obl \ph \to \neg \obl \neg \psi$ \\
$(\obl \ph \land \obl \psi) \to \obl (\ph \land \psi)$ & $(\obl \ph \land \Box (\ph \to \psi)) \to \obl \psi$
\end{tabular}
\end{center}



A straightforward derivation having a moral dilemma as assumption leads us to a contradiction in both cases. A long debate about the existence of moral dilemmas can be found in a number of papers, for example \cite{Bohse2005,Conee1982, Holbo2002, Marcus1980, McConnell1978, ToddWeber2002, Weber2007}. The position that  defends moral dilemmas do not exist, take moral principles to be accepted based on their intuitiveness. Given they are on a par with our usually accepted logical systems, rejecting them would go against our rationality. Expressing moral dilemmas formally, as can be seen above, lead to contradictions, thus, to trivialization, hence, their existence cannot be accepted together with the usual deontic principles (represented by Pair 1 and Pair 2 above). In response to this position, some authors defend a dialetheist account of moral systems, in order to accept moral dilemmas \cite{RibeiroTeles2021,  VazMaruchi25, Weber2007}. 

Formally, this dialetheist account of moral dillemas can be supported by the logic $\mathsf{CLuND}$, which is a deontic logic based on $\clun$\footnote{The system was named \textit{DPI} in \cite{Coniglio2009b,Coniglio2009, Peron2008} in reference to the work written by Batens in 1980 \cite{Batens1980a}.}. The reasons for such a choice can be read at length on \cite{VazMaruchi25}. The paper cited offers an overview of the Logics of Deontic Inconsistency, while keeping track of how each of these systems deal with moral dilemmas, and how can moral dilemmas be accepted in the systems without resulting in trivialization. The diagnosis of the paper is that all of the deontic paraconsistent systems, except $\mathsf{CLuND}$, fall prey to the same problem: since they have the consistency operator $\circ$ in their propositional fragments, they end up collapsing deontic and formal notions. Although formally it has no consequence, when we interpret the obligation operator as moral obligations, the interaction between $\obl$ and $\circ$ feels artificial and ad hoc. $\mathsf{CLuND}$, on the other hand, keeps the possibility of having such a consistency arising on the modal level, thus being intepreted already in deontic terms. Not only this avoids the above problem, but also allows for moral dilemmas to be expressible naturally in the system, while also not leading to trivialization.

Beyond the point of the deontic consistency, the main focus is on the fact that $\clun$ as a basis allows for a deontic system that is genuinely paraconsistent, being possible to add a consistency operator on the level of modality alone. This step is optional, however and could, in principle be left out of the construction procedure. The motivation to have modal logics based on  ${\clon}$ is then clear: $\clun$ is a genuinely paraconsistent logic, the modal version of which gives rise to a genuinely paraconsistent logic which interprets deontic notions in a natural way and allows for the expressibility of moral dilemmas, without leading to explosion; $\clun$ extends $\clon$ by the addition of the Law of Excluded Middle as an axiom. To give a general account of $\clun$ as a modal logic, we need then to be able to give $\clon$ a modal account, which then will serve the purposes of being applied to $\clun$, the system which we actually want to use in more practical scenarios.

The acceptance of the deontic systems based on $\clun$ also come with a high price: accepting that contradictions, at least those deriving from obligations, occur. Instead of looking for a solution to a deontic problem, the formalization of a deontic version of $\clun$ simply gives the possibility of moral dilemmas to exist qua unsolvable deontic situations. Being unsolvable, however, does not mean everything follows from them.

In \cite{DalOliNeg18, Drobyshevich2021} the authors presented a systematic and thorough study of several extensions of non-normal \textsf{FDE}-based modal logics by employing so-called bi-neighborhood semantics. The main difference to the approach presented in this paper is that here the rejection sets for the operators are not defined. Now, it is entirely possible, similar to a refinement in non-deterministic semantics, to recapture the rejection sets and thus putting the neighborhood frames defined in this paper in the context of {\sf FDE}-based modal logics. If furthermore verification and rejection sets are exhaustive and exclusive we can connect our results to classical neighborhood semantics. The above fact speaks in favor of having a semantics for a modal account of $\clon$. 

As far as future research on extensions of $\eclon$ go, one obvious and promising direction would be the investigation of extension by rules rather then axioms, whether it is possible to replace axioms by rules and vice versa. For example, it seems pretty straightforward to replace $\ocirc(\ph \wedge \ps) \to (\ocirc \ph \wedge \ocirc \ps)$ by the monotonicity rule  $(\ocirc\textsf{Mon})\ \ph\to\ps/\ocirc\ph\to\ocirc\ps$. The replacement, however, asks for a definition of a global consequence relation, which, despite its general flavor, asks for a change of philosophical interpretation of the symbols and their generality.

Finally, we just want to mention that all results will easily carry over if we expand the language with even more modal operators. As Section \ref{main_extensions} shows, dealing with mixed axioms is an easy task in the proposed framework, being only a matter of picking the right restrictions on the neighborhood functions and how they interact with each other. 

\section{Final remarks}\label{sec:final}

In this paper we proposed to give semantics for modal logics based on $\clon$ and its extensions, $\clun$ and $\clan$. Due to its lack of axioms for the $\neg$ symbol, the semantic treatment of $\neg$ in $\clon$ must be non-deterministic. In order to capture the behavior of an unary operator that is not truth-functionally determined, we ought to be able to leave the truth assignments open to any possibility. This poses a challenge for defining a semantic structure that could formally determine the notion of consequence of such logics. 

Using the machinery of bi-neighborhood semantics developed in \cite{DalOliNeg18, Drobyshevich2021} and comments on the behavior of $\clon$-negation in \cite{OmoriSkurt2021}, we were able to design a structure that is able to deal with $\clon$-negation and further push the boundaries of the work on logics with a weak negation symbol. Moreover, our framework is easily extendable to multimodal logics and even logics with other propositional operators, such as the $\lfi$s. It is also possible to extend this approach to logics that are usually formalized with $4$ truth values, such as the case with $\sf FDE$. 

Philosophically, it means that we have now a formal account that is flexible enough to deal with sentences that would be otherwise avoided by other systems. This logic is hyperintensional with respect to negation, but the framework here proposed allows us to give them a solid mathematical treatment, despite this feature. Philosophical works on dialetheism, moral dilemmas, contradictions and other non-standard philosophical positions are sure to able to profit from the developments we propose here.

\bibliographystyle{apalike}
\bibliography{generic}

\appendix

\section*{Technical Appendix}

\subsection*{Proof of Lemma \ref{lem_posext}}

\begin{proof}\leavevmode
    \begin{enumerate}
        \item Assume $w \in \ptru{\ocirc (\ph \land \ps)}_\mu$, then $\ptru{\ph \land \ps}_\mu \in N\pl_\ocirc (w)$. Hence $\ptru{\ph}_\mu \cap \ptru{\ps}_\mu \in N\pl_\ocirc (w)$. By \emph{augmentation} $\ptru{\ph}_\mu \in N\pl_\ocirc (w)$, $\ptru{\ps}_\mu \in N\pl_\ocirc (w)$. Hence, $w \in \ptru{\ocirc \ph}_\mu$ and $w \in \ptru{\ocirc \ps}_\mu$, which implies $w \in \ptru{\ocirc \ph}_\mu \cap \ptru{\ocirc \ps}_\mu = \ptru{\ocirc \ph \land \ocirc \ps}_\mu$.
        

        For the other direction, assume \emph{augmentation} does not hold, i.e. there are $X\pl,Y\pl$ such that $X\pl\cap Y\pl \in N\pl_\ocirc(w)$ but $X\pl,Y\pl \notin N\pl_\ocirc(w)$. Now define a model $\mu$, where $\ptru{p}_\mu = X\pl$ and $\ptru{q}_\mu=Y\pl$. An easy calculation shows that $w \notin \ptru{\ocirc (\ph \land \ps) \to (\ocirc \ph \land \ocirc \ps)}_\mu$.

        \item Assume $w \in \ptru{\ocirc \ph \land \ocirc \ps}_\mu = \ptru{\ocirc \ph}_\mu \cap \ptru{\ocirc \ps}_\mu$. So $w \in \ptru{\ocirc \ph}_\mu$ and $w \in \ptru{\ocirc \ps}_\mu$. Hence, $\ptru{\ph}_\mu \in N\pl_\ocirc(w)$ and $\ptru{\ps}_\mu \in N\pl_\ocirc(w)$. By \emph{closure by intersection}, $\ptru{\ph}_\mu  \cap \ptru{\ps}_\mu \in N\pl_\ocirc(w)$, which implies $\ptru{\ph \land \ps}_\mu \in N\pl_\ocirc(w)$. Therefore, $w \in \ptru{\ocirc (\ph \land \ps)}_\mu$. 

        The other direction is proven similar to 1.

        \item Assume $w \in \ptru{\ocirc (\ph \to \ps)}_\mu$ and $w \in \ptru{\ocirc \ph}_\mu$.
    Then $\ptru{\ph \to \ps}_\mu \in N\pl_\ocirc(w)$ and $\ptru{\ph}_\mu \in N\pl_\ocirc(w)$. By \emph{closure by Intersection}, $\ptru{\ph \to \ps}_\mu \cap \ptru{\ph}_\mu \in N\pl_\ocirc (w)$. By definition, $(\overline{\ptru{\ph}_\mu} \cup \ptru{\ps}_\mu) \cap \ptru{\ph}_\mu \in N\pl_\ocirc (w)$. Hence $(\overline{\ptru{\ph}_\mu} \cap \ptru{\ph}_\mu) \cup (\ptru{\ps}_\mu \cap \ptru{\ph}_\mu) \in N\pl_\ocirc(w)$. Since $\overline{\ptru{\ph}_\mu} \cap \ptru{\ph}_\mu = \varnothing$, we have $\ptru{\ps}_\mu \cap \ptru{\ph}_\mu \in N\pl_\ocirc(w)$. By \emph{augmentation}, $\ptru{\ps}_\mu \in N\pl_\ocirc(w)$. Therefore, $w \in \ptru{\ocirc \ps}_\mu$.

    The other direction is proven similar to 1.

        \item Assume $ w \in \ptru{\ocirc (\ph \lor \ps)}_\mu$. This means that $\ptru{\ph \lor \ps}_\mu = \ptru{\ph}_\mu \cup \ptru{\ps}_\mu \in N\pl_\ocirc(w)$. By \emph{primeness}, $\ptru{\ph}_\mu \in N\pl_\ocirc(w)$ or $\ptru{\ps}_\mu \in N\pl_\ocirc(w)$. Hence, $w \in \ptru{\ocirc \ph}_\mu$ or $w \in \ptru{\ocirc \ps}_\mu$. 

        The other direction is proven similar to 1.

    \item Assume $w \in \ptru{\ocirc \ph}_\mu$. Then $\ptru{\ph}_\mu \in N\pl_\ocirc(w)$, which, by \emph{reflexivity}, implies $w \in \ptru{\ph}_\mu$.

    For the other direction assume \emph{reflexivity} does not hold, i.e. there is an $X\pl \in N\pl_\ocirc(w)$ but $w \notin X\pl$. Now define a model $\mu$, where $\ptru{p}_\mu = X\pl$. An easy calculation shows that $w \notin \ptru{\ocirc p \to p}_\mu$.

    \item Assume $w \in \ptru{\ocirc \ph}_\mu$. Then $\ptru{\ph}_\mu \in N\pl_\ocirc(w)$. By \emph{transitivity}, $\set{w\prim \in W : \ptru{\ph}_\mu \in N\pl_\ocirc(w\prim)} \in N\pl_\ocirc(w)$, which means $\ptru{\ocirc \ph }_\mu \in N\pl_\ocirc(w)$. Therefore $w \in \ptru{\ocirc\ocirc \ph}_\mu$. 
    
    The other direction is proven similar to 5.
\qedhere
    \end{enumerate}
\end{proof}

\subsection*{Proof of Lemma \ref{lem_negext}}
\begin{proof}\leavevmode
 \begin{enumerate}
    \item 
    Assume $w \in \ptru{\ocirc (\neg \ph \wedge \neg \ps)}_\mu$, then $\ptru{\neg \ph \wedge \neg \ps}_\mu \in N\pl_\ocirc(w)$, i.e. $\ptru{\neg \ph}_\mu\cap\ptru{\neg \ps}_\mu \in N\pl_\ocirc(w)$. Hence $\mtru{\ph}_\mu\cap\mtru{\ps}_\mu \in N\pl_\ocirc(w)$. By the condition, we have $\mtru{\ph}_\mu,\mtru{\ps}_\mu \in N\mn_\ocirc(w)$. This implies $w\in \mtru{\ocirc \ph}_\mu \cap \mtru{\ocirc \ps}_\mu$, hence $w\in \ptru{\neg\ocirc \ph}_\mu \cap \ptru{\neg\ocirc \ps}_\mu$, but then $w\in\ptru{\neg \ocirc \ph \wedge \neg \ocirc \ps}_\mu$.

    For the other direction, assume the condition does not hold, i.e. there are $X\mn,Y\mn$ such that $X\mn\cap Y\mn \in N\pl_\ocirc(w)$ but $X\mn,Y\mn \notin N\mn_\ocirc(w)$. Now define a model $\mu$, where $\mtru{p}_\mu = X\mn$ and $\mtru{q}_\mu=Y\mn$. An easy calculation shows that $w \notin \ptru{\ocirc (\neg \ph \land \neg \ps) \to (\neg \ocirc \ph \land \neg \ocirc \ps)}_\mu$. 
    
    \item 
    Assume $w\in\ptru{\neg \ocirc \ph \wedge \neg \ocirc \ps}_\mu$, then $w\in \ptru{\neg\ocirc \ph}_\mu \cap \ptru{\neg\ocirc \ps}_\mu$, i.e. $w\in \mtru{\ocirc \ph}_\mu \cap \mtru{\ocirc \ps}_\mu$. This implies $\mtru{\ph}_\mu,\mtru{\ps}_\mu \in N\mn_\ocirc(w)$. By the condition, we have $\mtru{\ph}_\mu\cap\mtru{\ps}_\mu \in N\pl_\ocirc(w)$, i.e. $\ptru{\neg \ph}_\mu\cap\ptru{\neg \ps}_\mu \in N\pl_\ocirc(w)$, hence $\ptru{\neg \ph \wedge \neg \ps}_\mu \in N\pl_\ocirc(w)$. But then $w \in \ptru{\ocirc (\neg \ph \wedge \neg \ps)}_\mu$.

    The other direction is proven similar to 1.

    \item 
    Assume $w \in \ptru{\ocirc(\neg \ph \lor \neg \ps)}_\mu$ then $\ptru{\neg \ph \lor \neg \ps}_\mu \in N\pl_\ocirc$, hence $\ptru{\neg \ph}_\mu \cup \ptru{\neg \ps}_\mu  \in N\pl_\ocirc$, i.e. $\mtru{\ph}_\mu \cup \mtru{\ps}_\mu  \in N\pl_\ocirc$. By the condition, we have $\mtru{\ph}_\mu \in N\mn_\ocirc(w)$ or $\mtru{\ps}_\mu \in N\mn_\ocirc(w)$. Hence $w \in \mtru{\ocirc \ph}_\mu$ or $w \in \mtru{\ocirc \ps}_\mu$, i.e. $w \in \ptru{\neg \ocirc \ph}_\mu$ or $w \in \ptru{\neg \ocirc \ps}_\mu$. But this means $w \in \ptru{\neg \ocirc \ph}_\mu \cup \ptru{\neg \ocirc \ps}_\mu$, therefore $w \in \ptru{\neg \ocirc \ph \vee \neg \ocirc \ps}_\mu$.

    The other direction is proven similar to 1.

    \item 
    Assume $w \in \ptru{\neg \ocirc \ph}_\mu$, then $w \in \mtru{\ocirc \ph}_\mu$, hence $\mtru{\ph}_\mu \in N\mn_\ocirc(w)$. By the condition we have $w\in \mtru{\ph}_\mu$, hence $w\in \ptru{\neg \ph}_\mu$.
    
    For the other direction assume the condition does not hold, i.e. there is an $X\mn \in N\mn_\ocirc(w)$ but $w \notin X\mn$. Now define a model $\mu$, where $\mtru{p}_\mu = X\mn$. An easy calculation shows that $w \notin \ptru{\ocirc p \to p}_\mu$.
    \qedhere
\end{enumerate}   
\end{proof}

\subsection*{Proof of Lemma \ref{lem_intext}}

\begin{proof}\leavevmode
\begin{enumerate}
    \item 
    Assume $w \in \ptru{\ocirc \ph}_\mu$, then $\ptru{\ph}_\mu \in N\pl_\ocirc(w)$. By the condition we have $\ptru{\ph}_\mu \in N\pl_\bcirc(w)$, i.e. $w \in \ptru{\bcirc \ph}_\mu$.

    The other direction is proven similar to 2.

    \item 
    Assume $w \in \ptru{\neg \ocirc \ph}_\mu$, hence $w \in \mtru{\ocirc \ph}_\mu$ then $\mtru{\ph}_\mu \in N\mn_\ocirc(w)$. By the condition we have $\mtru{\ph}_\mu \in N\mn_\bcirc(w)$, i.e. $w \in \mtru{\bcirc \ph}_\mu$. Therefore $w \in \ptru{\neg \bcirc \ph}_\mu$.

    For the other direction assume the condition does not hold, i.e. there is an $X\mn \in N\mn_\ocirc(w)$ but $X\mn \notin N\mn_\bcirc(w)$. Now define a model $\mu$, where $\mtru{p}_\mu = X\mn$. An easy calculation shows that $w \notin \ptru{\neg \ocirc p \to \neg \bcirc p}_\mu$.

    \item 
    Assume $w \in \ptru{\neg \ocirc \ph}_\mu$, hence $w \in \mtru{\ocirc \ph}_\mu$ then $\mtru{\ph}_\mu \in N\mn_\ocirc(w)$. By the condition we have $\mtru{\ph}_\mu \in N\pl_\bcirc(w)$, i.e. $\ptru{\neg \ph}_\mu \in N\pl_\bcirc(w)$. Therefore $w \in \ptru{\bcirc \neg \ph}_\mu$.

    The other direction is proven similar to 2.

    \item 
    Assume $w \in \ptru{\ocirc \neg \ph}_\mu$, hence $\ptru{\neg \ph}_\mu \in N\pl_\ocirc(w)$ Then $\mtru{\ph}_\mu \in N\pl_\ocirc(w)$. By the condition we have $\mtru{\ph}_\mu \in N\mn_\bcirc(w)$, i.e. $w \in \mtru{\bcirc \ph}_\mu$. Therefore $w \in \ptru{\neg \bcirc \ph}_\mu$.
    
    The other direction is proven similar to 2. \qedhere
\end{enumerate}    
\end{proof}

\end{document}